\documentclass[11pt]{article}
\usepackage[letterpaper,margin=1in]{geometry}
\usepackage[T1]{fontenc}
\usepackage{lmodern,microtype}
\usepackage{amsmath,amssymb,amsthm,mathtools}
\usepackage{aliascnt,enumitem}
\usepackage{algorithm}
\usepackage{algpseudocode}
\usepackage{float}
\usepackage[backend=biber,style=numeric-comp,sorting=nyt,natbib=true,maxbibnames=99,maxcitenames=2,mincitenames=1]{biblatex}
\usepackage{xcolor}
\usepackage[most]{tcolorbox}
\usepackage[colorlinks=true,linkcolor=blue!45!black,citecolor=blue!45!black,urlcolor=blue!45!black]{hyperref}
\usepackage[nameinlink,noabbrev]{cleveref}
\setlist[itemize]{leftmargin=1.5em,itemsep=2pt,topsep=4pt}
\setlist[enumerate]{leftmargin=1.7em,itemsep=3pt,topsep=4pt}
\allowdisplaybreaks[1]
\newtcolorbox{questionbox}{
  colback=gray!6,
  colframe=gray!30,
  boxrule=0.4pt,
  arc=1.5mm,
  left=2mm,right=2mm,top=1.5mm,bottom=1.5mm,
  before skip=6pt,after skip=6pt
}

\newtheorem*{externaltheorem}{Theorem}
\newtheorem{theorem}{Theorem}[section]
\newaliascnt{lemma}{theorem}
\newtheorem{lemma}[lemma]{Lemma}
\aliascntresetthe{lemma}
\newaliascnt{corollary}{theorem}
\newtheorem{corollary}[corollary]{Corollary}
\aliascntresetthe{corollary}
\newaliascnt{proposition}{theorem}
\newtheorem{proposition}[proposition]{Proposition}
\aliascntresetthe{proposition}
\AddToHook{env/theorem/begin}{\crefalias{section}{theorem}}
\AddToHook{env/lemma/begin}{\crefalias{section}{lemma}}
\AddToHook{env/corollary/begin}{\crefalias{section}{corollary}}
\AddToHook{env/proposition/begin}{\crefalias{section}{proposition}}
\theoremstyle{definition}
\crefname{theorem}{Theorem}{Theorems}
\crefname{lemma}{Lemma}{Lemmas}
\crefname{corollary}{Corollary}{Corollaries}
\crefname{proposition}{Proposition}{Propositions}
\crefname{section}{Section}{Sections}
\newcommand{\E}{\mathbb E}
\newcommand{\Prb}{\mathbb P}
\newcommand{\1}{\mathbf 1}
\newcommand{\ALG}{\mathsf{ALG}}
\newcommand{\OPT}{\mathsf{OPT}}
\newcommand{\Ranking}{\textnormal{\textsc{Ranking}}}
\newcommand{\Scan}{\operatorname{Scan}}
\newcommand{\rhostar}{\rho^\star}
\hypersetup{pdftitle={Beating One Half for Online Bipartite Matching with Reusable Resources},pdfauthor={Xiaohui Bei, Zhihao Gavin Tang, Wenhao Wu}}

\title{Beating One Half for Online Bipartite Matching\\with Reusable Resources}
\author{%
Xiaohui Bei\thanks{Nanyang Technological University. Emails: \href{mailto:xhbei@ntu.edu.sg}{xhbei@ntu.edu.sg} (Xiaohui Bei) and \href{mailto:wenhao008@e.ntu.edu.sg}{wenhao008@e.ntu.edu.sg} (Wenhao Wu).}%
\and Zhihao Gavin Tang\thanks{Key Laboratory of Interdisciplinary Research of Computation and Economics, Shanghai University of Finance and Economics. Email: \href{mailto:tang.zhihao@mail.shufe.edu.cn}{tang.zhihao@mail.shufe.edu.cn}.}%
\and Wenhao Wu\footnotemark[1]%
}
\date{}

\begin{document}
\maketitle
\begin{abstract}
We study online bipartite matching with unit-inventory reusable resources, where requests arrive in an adversarially fixed order, and each use of a resource makes it unavailable for an independent duration drawn from a resource-dependent distribution. The benchmark knows all requests in advance but cannot observe a duration before choosing the corresponding use.

The classical \Ranking{} algorithm of Karp, Vazirani, and Vazirani (STOC 1990) fixes a uniformly random priority order of the resources and matches each arriving request to its highest-priority available neighbor. It achieves the optimal competitive ratio $1-1/e$ for unweighted nonreusable resources, but whether it beats $1/2$ for reusable resources has remained open. We prove that, for unweighted resources with resource-dependent stochastic durations, \Ranking{} achieves a competitive ratio of $(5-2\sqrt3)/3\approx0.511966$. We also give a black-box reduction from unweighted \Ranking{} to resource-weighted matching: any unweighted competitive ratio $\alpha>1/2$ yields a weighted ratio strictly above $1/2$. With independent sampling access to the duration distributions, the reduction gives a weighted ratio of $0.500034$. These results resolve two questions left open by Delong et al. (MOR 2024): whether \Ranking{} beats $1/2$, and whether one can beat $1/2$ under stochastic durations.

We analyze \Ranking{} resource by resource, rather than request by request. For deterministic durations, this gives a reduction to random-order greedy for a coverage function. 
We then extend the analysis to stochastic durations by comparing the residual schedules of \Ranking{} and a greedy algorithm, and apply a finer analysis of the random ranks to obtain the stated $0.511$ bound. For the weighted reduction, we apply \Ranking{} within groups of similar weights and uses weighted greedy to control the loss between groups.
\end{abstract}

\newpage

\section{Introduction}\label{sec:intro}

Many online platforms and service systems face a common operational problem: a limited pool of resources must be allocated to demand that arrives sequentially over time. Examples include car rental, job scheduling, hospital bed assignment, and cloud computing. In many such settings, using a resource does not consume it permanently: the resource is occupied for some time and then becomes available again. A vehicle can serve another trip after it returns, a server after a job is completed, and a room or bed after it is released. We refer to such resources as \emph{reusable}.

With reusable resources, an assignment affects both current reward and future availability. Serving a request now may prevent the same resource from serving a later request with fewer alternatives. Holding the resource back preserves flexibility, but may waste capacity if favorable future demand never arrives. The tradeoff is already prevalent with \emph{unit inventory}, where there is only one copy of each resource. Heterogeneous or stochastic service durations add another source of uncertainty because different uses release capacity at different times.

A natural model for such problems is \emph{online bipartite matching with reusable resources} (OBMRR). There is a fixed set $V$ of resources, and requests arrive one at a time. Each request must be rejected or assigned immediately to an available compatible resource. A use of resource $v$ earns reward $c_v$ and makes it unavailable for a duration drawn from $\mathcal D_v$. Durations are independent, and their distributions may depend on the resource. In the unweighted problem, all $c_v=1$; in the resource-weighted problem, rewards may differ arbitrarily. We compare the algorithm with a nonanticipatory benchmark: it knows the full request sequence, the compatibility graph and all duration distributions, but learns the duration of each use only after choosing that use.

If all durations exceed the time horizon, no resource returns before the last request, and the model reduces to classical online bipartite matching. In that setting, the celebrated \Ranking{} algorithm of \citet{karp-vazirani-vazirani-1990} samples one uniformly random priority order of the resources and assigns each request to its highest-priority available neighbor. It achieves the optimal competitive ratio $1-1/e$; \citet{devanur-jain-kleinberg-2013} give a primal-dual analysis. \Ranking{} also has a randomized static-pricing interpretation~\citep{eden-feldman-fiat-segal-2021}: assign each resource a random fixed price and match each request to its cheapest available compatible resource.

For the case that all the durations take the same constant, \citet{delong-et-al-2024} were the first to beat $1/2$ with integral algorithms at unit inventory. Their \textit{Periodic Reranking} algorithm achieves competitive ratio $0.589$, and an online-correlated-rental algorithm achieves $0.505$. Both guarantees extend to weighted resources.
Under the reusable setting, the difficulty of analyzing \Ranking{} is that changing one priority can affect assignments throughout the horizon. An early match changes a return time, which may change a later match and the return time of another resource. The resulting chain of changes is not controlled by the local comparisons used in classical \Ranking{} analyses. 

The above results lead to two questions left open by \citet{delong-et-al-2024}:

\begin{questionbox}
\textbf{Question 1.}\quad
Does \Ranking{} achieve a competitive ratio strictly larger than $1/2$ in the general reusable resources setting?
\end{questionbox}

\begin{questionbox}
\textbf{Question 2.}\quad
Can one beat $1/2$ for heterogeneous deterministic durations, or even for resource-dependent stochastic durations?
\end{questionbox}
The second question is also highlighted in recent overviews: \citet{huang-tang-wajc-2024} identify heterogeneous deterministic durations as an open case, while \citet{niazadeh-udwani-2026} discuss the general stochastic-duration problem as a prominent open direction.

\subsection{Our Results}

Our first result answers both questions for unweighted resources.

\begin{theorem}\label{thm:main}
For every unweighted instance with unit-inventory reusable resources and resource-dependent i.i.d.\ usage durations, the \Ranking{} algorithm satisfies
\begin{equation}\label{eq:main-guarantee}
    \E[\ALG]\ge \rhostar\E[\OPT],
    \qquad
    \rhostar=\frac{5-2\sqrt3}{3}=0.511966\ldots.
\end{equation}
The algorithm needs neither the duration distributions nor sampling access to them.
\end{theorem}

Taking each $\mathcal D_v$ to be a point mass includes heterogeneous deterministic durations. The benchmark may adapt to the durations of its previous uses; it is not restricted to a fixed offline schedule.

For weighted resources, we give a black-box reduction.

\begin{theorem}[Black-box reduction]\label{thm:weighted-reduction}
Suppose \Ranking{} is $\alpha$-competitive on all unweighted instances of the model, for some $\alpha>1/2$. There is a randomized online algorithm for the resource-weighted model with competitive ratio $1/2+\varepsilon(\alpha)$, where $\varepsilon(\alpha)>0$ depends only on $\alpha$. The algorithm uses \Ranking{} on unweighted subinstances and requires independent sampling access to each duration distribution.
\end{theorem}

The reduction uses the unweighted guarantee only through $\alpha$; it does not use the proof of that guarantee. Applying it to \cref{thm:main} gives the following numerical bound.

\begin{corollary}\label{cor:weighted-main}
For resource-weighted instances with unit inventory and resource-dependent i.i.d.\ usage durations, there is a randomized online algorithm satisfying
\begin{equation}\label{eq:weighted-guarantee}
    \E[\ALG]\ge 0.500034\,\E[\OPT].
\end{equation}
The algorithm requires independent sampling access to each duration distribution.
\end{corollary}

Sampling is used to simulate weighted greedy. A followed assignment from an inactive bucket uses its actual duration; other greedy uses receive independent samples. For deterministic durations, no sampling oracle is needed.

\subsection{Technical Overview}

\paragraph{From requests to resources.}
The usual primal-dual analysis of \Ranking{} follows requests in their arrival order~\citep{devanur-jain-kleinberg-2013}. We instead process resources in priority order. Each resource scans the requests left by higher-priority resources and accepts whenever it is available. This produces exactly the same matching, but avoids following the chain of changes caused by a local change in priority.

For deterministic durations, the new view gives a direct connection to online submodular maximization. A resource chooses one feasible schedule, and the objective is the number of requests covered by the chosen schedules. Choosing at most one schedule per resource is a partition matroid constraint. Moreover, the earliest feasible scan maximizes the number of uncovered requests served by a resource. Thus \Ranking{} is exactly random-order greedy for this coverage problem. The theorem of \citet{buchbinder-et-al-2020} immediately gives a $0.5096$ guarantee.

\paragraph{Stochastic durations.}
With stochastic durations, there is no fixed family of feasible schedules, so the deterministic coverage reduction is no longer literal. The key replacement is a residual schedule comparison: on any fixed set of remaining requests, the earliest feasible scan of a resource serves at least as many requests in expectation as the remaining part of any independent nonanticipatory schedule for that resource. This lets us keep the resource-by-resource view and the same coverage-style accounting. To improve on the resulting $1/2$ bound, we assign uniform ranks to resources and track benchmark requests relative to the ranks of their owners. A comparison across a rank cut gives one more inequality, and a continuous combination of the three inequalities yields $(5-2\sqrt3)/3$. The same argument also improves the $0.5104$ guarantee of \citet{korula-et-al-2018} for second-order-supermodular valuations; see \cref{app:submodular}.

\paragraph{Weighted resources.}
Traditional primal-dual analysis of matching can incorporate weighted resource rewards; in particular, the reusable-resource algorithms of \citet{delong-et-al-2024} retain their competitive ratio guarantees in the weighted setting. Our unweighted proof does not apply the primal-dual method, so it does not extend in this way. Instead, we use the unweighted guarantee as a black box. We group similar weights using randomly shifted geometric buckets~\citep{epstein-levin-mestre-segev-2010} and activate each bucket with a small probability. Active buckets run unweighted \Ranking{}; inactive resources follow an independent weighted-greedy execution. We filter active edges so that an active match can displace an inactive greedy match only when the latter has much smaller reward. The random shift controls the loss from separating similar weights, and sparse activation controls competition between buckets. The proof then combines the unweighted $\alpha$ guarantee with the $1/2$ guarantee of weighted greedy. The reduction loses part of the unweighted improvement, but preserves a ratio strictly above $1/2$. The idea of using filtered unweighted subproblems to improve a weighted algorithm also appears in \citet{gamlath-kale-mitrovic-svensson-2019}.

\subsection{Related Work}

\paragraph{Classical online matching.}
The optimality of \Ranking{} for unweighted nonreusable matching originates with \citet{karp-vazirani-vazirani-1990}; subsequent analyses include \citet{birnbaum-mathieu-2008,devanur-jain-kleinberg-2013}. \citet{aggarwal-goel-karande-mehta-2011} obtain the optimal $1-1/e$ ratio for resource-weighted nonreusable matching by perturbing rewards. Related extensions include budgeted allocation~\citep{mehta-saberi-vazirani-vazirani-2007}, edge-weighted matching~\citep{fahrbach-huang-tao-zadimoghaddam-2022}, and random-arrival matching~\citep{mahdian-yan-2011}. We refer to \citet{mehta-2013,huang-tang-wajc-2024} for surveys.

\paragraph{Reusable resources.}
Work on reusable resources considers several arrival models, inventory sizes, and allocation rules. In reusable assortment, \citet{rusmevichientong-sumida-topaloglu-2020} study random usage durations, while \citet{gong-et-al-2022} allow potentially adversarial customer preferences and resource-dependent stochastic usage times; the latter obtain a $1/2$ guarantee for a myopic policy and also show that request-dependent usage-time distributions can rule out any nontrivial competitive guarantee. In a broader online-allocation framework encompassing matching, budgeted allocation, and assortment, \citet{goyal-iyengar-udwani-2025} allow arbitrary resource-dependent stochastic usage distributions and obtain the optimal $1-1/e$ ratio when every resource has large starting inventory. \citet{feng-niazadeh-saberi-2024} study Bayesian reusable assortment with large inventory, and \citet{feng-niazadeh-saberi-replenishment} study robustness to exogenous and allocation-coupled inventory shocks. At unit inventory, \citet{delong-et-al-2024} give the common-duration guarantees discussed above. We allow resource-dependent i.i.d.\ durations and prove that \Ranking{} beats $1/2$ for unweighted resources. \citet{baek-wang-2026} study weighted greedy under resource-dependent stochastic durations and prove ratios $1/(2-\beta)$ in general and $(1+\beta)/2$ for geometric durations, where $\beta$ is the minimum probability of an immediate return; their guarantees improve on $1/2$ when $\beta$ is bounded away from zero, whereas ours are uniform over all resource-dependent i.i.d.\ duration distributions. Related reusable matching and allocation models include \citet{dickerson-et-al-2018,dickerson-et-al-2019} and the network-resource framework of \citet{simchi-levi-zheng-zhu-2025}. Other directions include pricing and revenue management~\citep{besbes-elmachtoub-sun-2022,balseiro-ma-zhang-2025,lei-jasin-2020}, online learning and nonstationary allocation~\citep{jia-shi-shen-2024,zhang-cheung-2025}, loss-network allocation~\citep{xie-gurvich-kucukyavuz-2025}, and fair allocation~\citep{liu-hajiesmaili-2025}.

\paragraph{Random-order submodular maximization.}
For adversarial-order online submodular welfare, greedy's $1/2$ guarantee is optimal in the standard efficient-algorithm setting~\citep{kapralov-post-vondrak-2013}. In random order, \citet{korula-et-al-2018} prove ratios $0.5052$ for general submodular welfare and $0.5104$ for second-order-supermodular valuations, including coverage valuations. \citet{buchbinder-et-al-2020} give a simpler $0.5096$ analysis that also applies when the arriving objects are the classes of a partition matroid. Our deterministic reduction uses this latter model. Our rank-based analysis also applies to second-order-supermodular welfare and improves the $0.5104$ guarantee to $(5-2\sqrt3)/3$; see \cref{app:submodular}.

\subsection{Organization}
\Cref{sec:model} defines the model. \Cref{sec:schedules} gives the deterministic warm-up. \Cref{sec:stochastic} develops the stochastic-duration analysis and proves the unweighted guarantee, while \cref{sec:weighted} gives the reduction to weighted resources. \Cref{app:combine,app:arithmetic} contain the calculations for the two competitive ratios. \Cref{app:submodular} extends the rank analysis to second-order-supermodular welfare.

\section{Model and Preliminaries}\label{sec:model}

\paragraph{Instance and dynamics.}
There is a finite set $V$ of resources and a time horizon $[T]=\{1,\ldots,T\}$. Resource $v$ has reward (or weight) $c_v\ge 0$ and is compatible with the requests in a set $N_v\subseteq[T]$. The request sequence and the compatibility graph are fixed in advance by an oblivious adversary. When request $t$ arrives, the algorithm learns its compatible resources and must immediately either assign it to one currently available resource or reject it. Every resource has unit inventory and is initially available; assigning request $t$ to resource $v$ earns reward $c_v$.

Successive uses of a resource have i.i.d. durations. Specifically, for each $v\in V$,
\[
    d_{v,1},d_{v,2},\ldots\stackrel{\mathrm{i.i.d.}}{\sim}\mathcal D_v,
\]
where the duration sequences are independent across resources and take values in the positive integers. If $v$ is used at time $t$ and the realized duration is $d$, then $v$ is unavailable at times $t+1,\ldots,t+d-1$ and becomes available again at time $t+d$. The distribution $\mathcal D_v$ may depend on the resource, but not on the request or on the use number. Deterministic duration $d_v$ is the special case in which $\mathcal D_v$ is a point mass.

\paragraph{Information and benchmark.}
A policy is \emph{nonanticipatory} if every assignment decision is made before the duration of that use is observed; in particular, it cannot condition on duration samples from future uses. After a use is chosen, its realized duration is revealed, so both the online algorithm and the benchmark may adapt to durations observed from earlier uses.

Our benchmark is the optimal nonanticipatory policy with full knowledge of the request sequence, compatibility graph, rewards, and duration distributions. It must nevertheless process requests in chronological order and, like the online algorithm, learns the duration of a use only after committing to that use. Let $\OPT$ denote its random total reward and write
\[
    M=\E[\OPT].
\]
The online algorithm's random reward is denoted by $\ALG$. We say that an algorithm is $\rho$-competitive if $\E[\ALG]\ge \rho M$ for every instance. In the unweighted setting, $c_v=1$ for every resource $v$.

For technical convenience, durations may be truncated at $T+1$: values larger than $T+1$ have the same effect on all future availability decisions within the horizon and reveal no information about later i.i.d. durations. After this truncation the state and action spaces are finite, so an optimal nonanticipatory benchmark exists by backward induction.

\paragraph{Independence conventions.}
Unless a coupling is stated explicitly, the main algorithm, the benchmark, and any auxiliary executions use independent priority randomization and independent duration samples. Whenever we remove requests or edges in the analysis, the remaining requests retain their original arrival times.

\paragraph{The \Ranking{} algorithm.}
Before the first request arrives, \Ranking{} samples a uniformly random permutation of the resources. Each request is assigned to its highest-priority available compatible resource, if one exists. Equivalently, assign independent ranks
\[
    U_v\sim\mathrm{Unif}[0,1],
\]
where smaller rank means higher priority, and always choose an available compatible resource of minimum rank.

\section{Warm-Up: Deterministic Durations}\label{sec:schedules}

We first consider the instances with deterministic durations and unweighted resources. In this special case, a resource-by-resource view of execution turns \Ranking{} into a greedy algorithm for a coverage problem studied by \citet{buchbinder-et-al-2020}. This gives a direct $0.5096$ guarantee and motivates the analysis used for stochastic durations in the next section.

\subsection{Resource-by-Resource Execution}

For a resource $v$ and a request set $B\subseteq[T]$, let $\Scan_v(B)$ be the schedule obtained by scanning $B\cap N_v$ in chronological order: starting with $v$ available and accepting whenever $v$ is available, using an i.i.d. duration sequence from $\mathcal D_v$. Thus $\Scan_v(B)$ is random in general and is deterministic once the duration sequence is fixed. In this section, $\mathcal D_v$ is a point mass, so $\Scan_v(B)$ is deterministic.

Now fix a resource order and all duration sequences. Process resources from highest to lowest priority. For each resource, use the same fixed duration sequence as in chronological fixed-priority greedy. When resource $v$ is reached, let it scan the requests not already assigned. If $S_v$ is the resulting schedule, then
\begin{equation}\label{eq:resource-ordered}
    S_v=\Scan_v\left([T]\setminus\bigcup_{u\text{ before }v}S_u\right).
\end{equation}

\begin{lemma}[Resource-by-resource execution]\label{lem:resource-ordered}
For every priority order and realization of the durations, \eqref{eq:resource-ordered} gives the same matching as chronological fixed-priority greedy.
\end{lemma}
\begin{proof}
Induct on the resource order. Once the schedules of higher-priority resources agree, the two executions offer $v$ the same remaining requests in the same arrival order. They also use the same duration for each successive use of $v$, so its acceptance times agree.
\end{proof}

Note that this is a different view to analyze the algorithm, not a different online implementation. The lemma also holds for stochastic durations after their values are fixed.

\subsection{Reduction to Coverage Maximization}

Suppose resource $v$ has deterministic duration $d_v$. Its feasible schedules are
\[
    \mathcal S_v=\{F\subseteq N_v:
        |s-t|\ge d_v\text{ for all distinct }s,t\in F\}.
\]
The earliest feasible scan is optimal for a single resource, even after some requests have been removed.

\begin{lemma}[Earliest feasible schedule]\label{lem:det-greedy}
For every resource $v$ and request set $B$,
\[
    |\Scan_v(B)|=\max\{|F|:F\in\mathcal S_v,\ F\subseteq B\}.
\]
\end{lemma}
\begin{proof}
Let $x_r$ and $y_r$ be the $r$th requests in the scan and in any feasible schedule in $B$, respectively. We show by induction that $x_r$ exists and $x_r\le y_r$ whenever $y_r$ exists. This holds for the first request. If it holds for $r$, then
\[
    y_{r+1}\ge y_r+d_v\ge x_r+d_v,
\]
so the scan can make its next use no later than $y_{r+1}$. It therefore serves at least as many requests as the other schedule.
\end{proof}

To obtain the coverage problem, make each pair $(v,F)$, with $F\in\mathcal S_v$, an element. All schedules for $v$ belong to one class
\[
    \mathcal P_v=\{(v,F):F\in\mathcal S_v\}.
\]
Choose at most one element from each class, and maximize
\[
    f(\mathcal X)=\left|\bigcup_{(v,F)\in\mathcal X}F\right|.
\]
Thus $f$ is a monotone coverage function, and the constraint is a partition matroid.

\begin{proposition}[Deterministic reduction]\label{thm:det-lifting}
The coverage problem has the same optimum as the deterministic matching problem. For any resource order, fixed-priority greedy is maximum-marginal greedy on the corresponding classes, with ties resolved by the earliest feasible scan.
\end{proposition}
\begin{proof}
A matching gives one schedule per resource and hence a coverage solution of the same value. Conversely, if chosen schedules overlap, keep each covered request in just one of them. Deleting requests preserves feasibility, so the remaining schedules form a matching with the same value.

Suppose the earlier classes cover $C$ and the next class is $\mathcal P_v$. Choosing $(v,F)$ adds $|F\setminus C|$ to the coverage. Since $F\setminus C$ is a feasible schedule on the remaining requests, \cref{lem:det-greedy} gives
\[
    \max_{F\in\mathcal S_v}|F\setminus C|
       =|\Scan_v([T]\setminus C)|.
\]
The resource-by-resource execution therefore makes a maximum-marginal choice in every class.
\end{proof}

We use the following theorem of Buchbinder et al.

\begin{externaltheorem}[Buchbinder et al.~\cite{buchbinder-et-al-2020}, Theorem~1]
Let $f$ be a nonnegative monotone submodular function on a ground set partitioned into classes $\mathcal P_1,\ldots,\mathcal P_m$. If the classes are processed in a uniformly random order and greedy chooses a maximum-marginal element from each class, then its expected value is at least $0.5096$ times the optimum subject to choosing at most one element from each class.
\end{externaltheorem}

This directly gives the following guarantee for \Ranking{} with deterministic durations.

\begin{corollary}\label{cor:det-5096}
For arbitrary resource-dependent deterministic durations, \Ranking{} is $0.5096$-competitive.
\end{corollary}
\begin{proof}
Under \Ranking{}, the resources, and hence the classes $\mathcal P_v$, are processed in a uniformly random order. By \cref{thm:det-lifting}, the execution makes a maximum-marginal choice in each class, and the lifted problem has the same optimum as the deterministic matching problem. Since $f$ is a monotone submodular coverage function, the theorem above gives the $0.5096$ guarantee.
\end{proof}

Notice that \Ranking{} does not enumerate the schedules in $\mathcal S_v$: by \cref{lem:det-greedy}, the earliest feasible scan makes the required greedy choice. The random order is over resource classes, not requests, and requests remain in adversarial order. Thus we use the partition-matroid result rather than a theorem that assumes random request arrivals.

\section{Stochastic Durations}\label{sec:stochastic}

We now return to resource-dependent stochastic durations and prove \cref{thm:main}. The coverage reduction from \cref{sec:schedules} no longer applies literally: before the duration samples are realized, there is no fixed family of feasible schedules from which a resource chooses. The resource-by-resource representation, however, remains valid after the samples are fixed, and it suggests the right substitute for the deterministic argument.

Our analysis of \Ranking{} with stochastic durations has three layers. First, we show that on any independently specified set of remaining requests, the earliest feasible scan of a resource is at least as good in expectation as the residual part of any independent nonanticipatory schedule. This is the stochastic analogue of \cref{lem:det-greedy}. Second, we use this comparison to obtain two rank-by-rank inequalities by charging benchmark requests to their benchmark \emph{owners}. These inequalities recover the familiar $1/2$ baseline and already improve on it when sufficiently many benchmark requests are captured after their owners are processed. Third, we introduce a cut in the random priority order. The cut recovers additional information about early captures and supplies the extra inequality needed to beat $1/2$. The purely analytic combination of the three inequalities is deferred to \cref{app:combine}.

\subsection{Residual Schedule Comparison}

The following expected comparison is the stochastic analogue of \cref{lem:det-greedy}. After removing an independently chosen set of requests, the earliest feasible scan performs at least as well in expectation as the part of any independent nonanticipatory schedule that remains.

\begin{lemma}[Residual schedule comparison]\label{lem:residual-comparison}
Let $\widetilde S_v$ be the schedule of resource $v$ under a nonanticipatory comparison policy.
Let $C$ be a set of requests independent of that execution and of the scan's duration samples.
Then
\begin{equation}\label{eq:residual-comparison}
    \E|\Scan_v([T]\setminus C)|
       \ge \E|\widetilde S_v\setminus C|.
\end{equation}
The same inequality holds after conditioning on external randomness that fixes $C$, provided the
comparison execution and the scan's duration samples retain the required independence.
\end{lemma}

\begin{proof}
First fix $C$. Draw i.i.d.\ samples
\[
    \xi_1,\xi_2,\ldots\sim\mathcal D_v
\]
for the scan. We couple these samples with the comparison execution in the following way: whenever the
comparison policy uses $v$ at a request outside $C$ for the $r$th time, give that use duration
$\xi_r$. Uses of $v$ at requests in $C$ receive fresh independent samples from $\mathcal D_v$.
All other duration samples in the comparison execution are generated independently as usual.

This coupling does not change the distribution of the comparison execution. Indeed, at every use
of $v$, whether the current request lies in $C$ is already determined before its duration is revealed.
Thus, conditional on the history before that duration is observed, the sample assigned to the use is
fresh, independent of the history, and distributed according to $\mathcal D_v$. By induction over
the successive decisions, the comparison policy therefore has the same distribution as in the
original model.

Let $\sigma_r$ be the time of the scan's $r$th acceptance, and let $\tau_r$ be the time of the
comparison policy's $r$th use of $v$ outside $C$. We claim that
\[
    \sigma_r\le\tau_r
\]
whenever $\tau_r$ exists. The claim is immediate for $r=1$. If it holds for $r$ and
$\tau_{r+1}$ exists, then the use at $\tau_r$ has duration $\xi_r$, so
\[
    \tau_{r+1}\ge\tau_r+\xi_r\ge\sigma_r+\xi_r.
\]
Moreover, $\tau_{r+1}$ is a request outside $C$ compatible with $v$. Hence it is available to the
scan after its $r$th busy period, and the scan accepts its next request no later than
$\tau_{r+1}$. Thus $\sigma_{r+1}\le\tau_{r+1}$.

Consequently, on every realization of this coupling,
\[
    |\Scan_v([T]\setminus C)|
       \ge |\widetilde S_v\setminus C|.
\]
Since the coupled comparison execution has the same distribution as the original one, taking
expectations proves \eqref{eq:residual-comparison}. For random $C$, condition first on $C$.
The conditional version follows in the same way.
\end{proof}

The rest of this section uses \cref{lem:residual-comparison} directly. It preserves the comparison underlying the deterministic coverage reduction without requiring a fixed family of stochastic schedules. The weighted analysis in \cref{sec:weighted} will use the same coupling again, both for weighted greedy and for restricted benchmark policies.

\subsection{Benchmark Owners and Rank Densities}\label{sec:owners}

We now use the residual schedule comparison to relate \Ranking{}'s gains to the benchmark requests that remain unserved as resources are processed in rank order. By distinguishing requests served before their benchmark owners from those served at or after their owners, we obtain a $1/2$ guarantee and identify a source of improvement beyond it.

For the remainder of this section, all rewards equal one. We run an optimal benchmark with randomness independent of the main execution and its ranks. Let $O_v$ be the requests served by resource $v$ in this benchmark, and put
\[
    O=\bigcup_{v\in V}O_v,
    \qquad M=\E|O|>0.
\]
The sets $O_v$ are pairwise disjoint. We call $v$ the \emph{owner} of every request in $O_v$.

Generate the main priority order by independent $U_v\sim\mathrm{Unif}[0,1]$ and define
\[
    P_x=\{v:U_v\le x\},\qquad Q_x=\{v:U_v>x\},
    \qquad O_P=\bigcup_{v\in P}O_v\quad(P\subseteq V).
\]
For a resource $v$, we write
\[
    O_{>v}=\bigcup_{u:U_u>U_v}O_u,
    \qquad
    O_{\le v}=\bigcup_{u:U_u\le U_v}O_u.
\]
Let $S_v$ be resource $v$'s schedule in the main resource-by-resource execution, and let
\[
    C_x=\bigcup_{v:U_v\le x}S_v,
    \qquad C_{<v}=\bigcup_{u:U_u<U_v}S_u.
\]
Thus $C_x$ is the set of requests served by resources with rank at most $x$.

The owner accounting separates benchmark requests according to whether they are captured before or after their owners are processed. This distinction matters because a request captured before its owner gives an immediate gain, but removes one request from the benchmark schedule that will later be used as a comparator for that owner. A request captured at or after its owner creates no such later loss. Accordingly, we track three expected counts: all requests served, benchmark requests served before their owners, and benchmark requests served at or after their owners. Normalize by $M$, and let $w,a,b$ be their densities with respect to rank. Formally, for $0\le x\le1$,
\begin{align*}
    \int_0^x w(s)\,ds
       &=\frac1M\E|C_x|,\\
    \int_0^x a(s)\,ds
       &=\frac1M\E\sum_{v:U_v\le x}|S_v\cap O_{>v}|,\\
    \int_0^x b(s)\,ds
       &=\frac1M\E\sum_{v:U_v\le x}|S_v\cap O_{\le v}|.
\end{align*}
These cumulative expectations are Lipschitz. Indeed, when the rank threshold increases from $x$ to $y$, only resources with ranks in $(x,y]$ are newly added, and each contributes at most $T$ requests. The expected increment is therefore at most $|V|T(y-x)$. Thus the densities exist and are bounded. In particular,
\begin{equation}\label{eq:W-density}
    W:=\int_0^1w(x)\,dx=\frac{\E[\ALG]}M.
\end{equation}
Only the early-capture density $a$ can reduce a benchmark schedule used later as a comparator. To account for this future loss, define
\begin{equation}\label{eq:J-def}
    J(x)=\int_0^x\frac{a(s)}{1-s}\,ds,\qquad x<1.
\end{equation}
The factor $1/(1-s)$ spreads a request captured at rank $s$ over its owner's possible ranks in $(s,1)$. Thus $J(x)$ is the normalized benchmark mass already lost from the comparator of a resource whose rank is $x$. All statements about densities below hold almost everywhere.

For the conditioning arguments below, expose the main execution in increasing rank order. By rank $x$, the ranks and schedules of resources in $P_x$, together with the duration samples used in those schedules, have been revealed. The exact ranks and all duration samples of resources in $Q_x$ remain unrevealed. We will use the identity
\begin{equation}\label{eq:deferred-rank}
    \Prb[U_v>y\mid U_v>x]=\frac{1-y}{1-x},
    \qquad x<y<1.
\end{equation}
Thus, after conditioning on the information revealed through rank $x$, the ranks of the resources in $Q_x$ are still independent and uniform on $(x,1)$. Since the benchmark is independent of the ranks, conditioning on the benchmark does not reveal the ranks of the unprocessed resources.

The next two inequalities formalize this accounting. The first simply records benchmark requests captured at the current rank. The second compares the current resource with its benchmark schedule and subtracts exactly the benchmark mass already captured from owners that were still unprocessed.

\begin{lemma}\label{lem:owner-inequalities}
For almost every $x\in[0,1]$,
\begin{align}
    w(x)&\ge a(x)+b(x),\label{eq:owner-capture}\\
    w(x)&\ge 1-J(x).\label{eq:owner-residual}
\end{align}
\end{lemma}
\begin{proof}
The first inequality follows because the two sets counted by $a$ and $b$ are disjoint subsets of the main schedule.

For the second inequality, we use a test-function argument to turn the resource-wise comparison into a bound on the rank density. Fix a bounded nonnegative measurable function $\varphi$ on $[0,1]$.

Consider a resource $v$, and condition on its rank $U_v$ and on everything revealed before $v$ is processed. Then $C_{<v}$ is fixed, while no duration sample of $v$ has been revealed, and the benchmark remains independent. Since
\[
    S_v=\Scan_v([T]\setminus C_{<v}),
\]
\cref{lem:residual-comparison} gives, under this conditioning,
\[
    \E[|S_v|]
    \ge
    \E[|O_v\setminus C_{<v}|].
\]
Moreover, $\varphi(U_v)$ is now fixed and nonnegative. Multiplying by it and then averaging over the conditioning gives
\[
    \E\bigl[\varphi(U_v)|S_v|\bigr]
    \ge
    \E\bigl[\varphi(U_v)|O_v\setminus C_{<v}|\bigr].
\]

Summing over $v$ and dividing by $M$, we obtain
\[
    \frac1M\E\sum_v\varphi(U_v)|S_v|
    \ge
    \frac1M\E\sum_v\varphi(U_v)|O_v|
    -\frac1M\E\sum_v\varphi(U_v)|O_v\cap C_{<v}|.
\]
By the definition of $w$, the left-hand side is
\[
    \int_0^1\varphi(x)w(x)\,dx.
\]
Since the benchmark is independent of the ranks and each $U_v$ is uniform on $[0,1]$,
\[
    \frac1M\E\sum_v\varphi(U_v)|O_v|
    =\int_0^1\varphi(x)\,dx.
\]
Hence
\begin{equation}\label{eq:test-owner}
    \int_0^1\varphi(x)w(x)\,dx
    \ge
    \int_0^1\varphi(x)\,dx-\mathcal L(\varphi),
\end{equation}
where
\[
    \mathcal L(\varphi)
    =
    \frac1M\E\sum_v
    \varphi(U_v)|O_v\cap C_{<v}|.
\]

It remains to identify this loss term. Every request counted in $\mathcal L(\varphi)$ was captured at some earlier rank $s$ while its owner was still unprocessed, and is therefore counted by $a(s)$. Conditional on the information revealed through rank $s$, the owner's rank is uniform on $(s,1)$. Thus its expected $\varphi$-weight is
\[
    \frac1{1-s}\int_s^1\varphi(x)\,dx,
\]
and therefore
\begin{align*}
    \mathcal L(\varphi)
    &=\int_0^1\frac{a(s)}{1-s}
        \left(\int_s^1\varphi(x)\,dx\right)ds\\
    &=\int_0^1\varphi(x)J(x)\,dx.
\end{align*}
Substituting this into \eqref{eq:test-owner} gives
\[
    \int_0^1
    \varphi(x)\bigl(w(x)-1+J(x)\bigr)\,dx
    \ge0
\]
for every bounded nonnegative measurable $\varphi$. Hence
\[
    w(x)\ge1-J(x)
\]
for almost every $x$ (otherwise take $\varphi$ to be the indicator of a set of positive measure on which the inequality fails). This proves \eqref{eq:owner-residual}.
\end{proof}

\begin{corollary}\label{cor:baseline}
Writing $\beta=\int_0^1b(x)\,dx$, we have
\begin{equation}\label{eq:baseline}
    W\ge\frac{1+\beta}{2}.
\end{equation}
\end{corollary}
\begin{proof}
Integrating the two inequalities in \cref{lem:owner-inequalities} gives
\[
    W\ge\int_0^1a(s)\,ds+\beta,
    \qquad
    W\ge 1-\int_0^1a(s)\,ds.
\]
Average the two bounds.
\end{proof}

If $\beta=\int_0^1 b(x)\,dx$ is large, \cref{cor:baseline} already gives a strict improvement. The difficult case is therefore when $\beta$ is small. In that case, the owner accounting by itself loses most of the benefit from requests captured before their owners, because those gains are later charged back to the owners' benchmark schedules. The next subsection recovers part of this information by using where in the ranking these early captures occur.

\subsection{Comparison across a Rank Cut}\label{sec:rank-cut}

Our next goal is to strengthen the preceding bounds when few benchmark requests are served at or after their owners. We compare the main execution with an independent execution on resources beyond a rank threshold, and obtain an additional lower bound on the main execution's gains after that threshold.

Fix $p\in(0,1)$ and set $q=1-p$. By rank $p$, the main execution has processed all resources in $P_p$, including the owners of the requests in $O_{P_p}$. Thus, any request in $O_{P_p}$ captured after the cut is a late capture and creates no future loss for its owner. To estimate how much of this benchmark mass is still reachable, run an independent \Ranking{} execution on the remaining resources $Q_p$ and ask how many requests in $O_{P_p}\setminus C_p$ it can serve.

The quantity $qD(p)$ below is the lower bound obtained by taking the auxiliary coverage of $O_{P_p}$ and subtracting the part already covered by $C_p$. Define
\begin{equation}\label{eq:D-def}
    D(p)=\frac1q\left(
      pJ(q)-\int_0^p\frac{p-s}{1-s}a(s)\,ds
             -\int_0^p b(s)\,ds\right).
\end{equation}
Here $pJ(q)$ will be the normalized auxiliary coverage of $O_{P_p}$, while the two integrals will be the normalized amount of $O_{P_p}$ already contained in $C_p$.

\begin{lemma}\label{lem:rank-cut}
For every $p\in(0,1)$ and almost every $x\in[p,1)$,
\begin{equation}\label{eq:rank-cut}
    w(x)\ge D(p)-\int_p^x\frac{b(s)}{1-s}\,ds.
\end{equation}
\end{lemma}
\begin{proof}
Write $P=P_p$ and $Q=Q_p$. On the resources in $Q$, run an independent copy of \Ranking{} with an independent uniform order and independent duration samples. Let $\widehat S_v$ be its schedule for $v\in Q$ and let $\widehat C=\bigcup_{v\in Q}\widehat S_v$.

\emph{Useful auxiliary coverage at the cut.}
We first estimate how many benchmark requests owned by $P$ are served by the auxiliary execution, and then subtract those already in $C_p$. To express the first quantity using the rank densities $a$ and $b$, we embed the auxiliary execution on $Q$ into the first $q$ fraction of a fresh full \Ranking{} order. Each resource belongs to $Q$ independently with probability $q$. Give resources in $Q$ independent ranks $R_v\sim\mathrm{Unif}[0,1]$ for the auxiliary order, give resources in $P$ independent ranks $R'_v\sim\mathrm{Unif}[0,1]$, and define
\[
    \widetilde U_v=
    \begin{cases}
        qR_v,&v\in Q,\\
        q+pR'_v,&v\in P.
    \end{cases}
\]
The variables $\widetilde U_v$ are independent uniforms on $[0,1]$. Conditional on $Q$, their relative order on $Q$ is uniform, and by \cref{lem:resource-ordered} the auxiliary execution depends only on this order and the duration samples of resources in $Q$. Complete the construction by assigning independent duration sequences to the resources in $P$. The resulting full \Ranking{} execution has the same law as the main execution. Its resources with new rank at most $q$ produce exactly the auxiliary execution, and $P$ is exactly the set of resources whose new rank exceeds $q$. Since this full execution has the same law as the main execution, its capture densities are again $a$ and $b$. A benchmark request counted by $a(s)$ at rank $s\le q$ has an owner whose rank is uniform on $(s,1)$, so the owner lies after $q$ with conditional probability $p/(1-s)$. A request counted by $b(s)$ already has its owner at rank at most $s\le q$ and therefore contributes zero. Hence
\begin{equation}\label{eq:aux-OP}
    \frac1M\E|O_P\cap\widehat C|
    =p\int_0^q\frac{a(s)}{1-s}\,ds=pJ(q).
\end{equation}
For $C_p$, every request counted by $b(s)$ at $s\le p$ belongs to $O_P$. If it is counted by $a(s)$, its owner is uniform on $(s,1)$ and lies in $P$ with conditional probability $(p-s)/(1-s)$. Hence
\begin{equation}\label{eq:main-OP}
    \frac1M\E|O_P\cap C_p|
    =\int_0^p b(s)\,ds+
      \int_0^p\frac{p-s}{1-s}a(s)\,ds.
\end{equation}
Using
\[
    |(O_P\setminus C_p)\cap\widehat C|
      \ge |O_P\cap\widehat C|-|O_P\cap C_p|,
\]
we obtain
\begin{equation}\label{eq:Y-initial}
    \frac1M\E|(O_P\setminus C_p)\cap\widehat C|
    \ge qD(p).
\end{equation}

\emph{Turning auxiliary coverage into main gain.}
Past the cut, we keep only auxiliary assignments whose resource has not yet been processed in the main order, since only those resources can still serve as comparators for the current main resource. For $x\ge p$, define
\[
    Y(x)=\frac1M\E\left|
      (O_P\setminus C_x)\cap
      \bigcup_{v:U_v>x}\widehat S_v\right|.
\]
Thus $Y(x)$ is the normalized expected amount of useful auxiliary coverage still available at rank $x$: the request is unserved by the main execution, and its auxiliary resource is still unprocessed in the main order. Equation~\eqref{eq:Y-initial} gives $Y(p)\ge qD(p)$.

The quantity $Y(x)$ is a remaining mass, whereas $w(x)$ is a gain density at rank $x$. Since every unprocessed resource has conditional rank density $1/(1-x)$ on $(x,1)$, the natural comparison rate is $Y(x)/(1-x)$. We now formalize this comparison.

Condition on $Q$ and on the information revealed by the main execution through rank $x$. For each unprocessed resource $v$, its exact main rank is still uniform on $(x,1)$, and none of its main duration samples has been revealed. Conditional on $Q$, the auxiliary execution is independent of this revealed information; we do not condition on its realized durations. By \cref{lem:residual-comparison}, the expected gain from processing $v$ next is at least $\E|\widehat S_v\setminus C_x|$ under this conditioning, and hence at least the contribution from $O_P\setminus C_x$. Summing over unprocessed resources, the restricted auxiliary contributions add up to exactly $M Y(x)$ because the schedules $\widehat S_v$ are pairwise disjoint. Multiplying by the conditional rank density $1/(1-x)$, dividing by $M$, and averaging gives

\begin{equation}\label{eq:gain-from-Y}
    w(x)\ge \frac{Y(x)}{1-x}.
\end{equation}
Equivalently, \eqref{eq:gain-from-Y} can be integrated against any bounded nonnegative measurable function supported in $(p,1)$.

\emph{Loss as the rank increases.}
The mass $Y(x)$ decreases even if the main execution serves no additional useful request, simply because auxiliary resources are gradually processed in the main order. Dividing by $1-x$ removes this proportional loss. We now record the remaining loss over a finite rank interval. For every $p\le x<y<1$,
\begin{equation}\label{eq:Y-evolution}
    \frac{Y(y)}{1-y}
    \ge \frac{Y(x)}{1-x}
           -\int_x^y\frac{b(s)}{1-s}\,ds.
\end{equation}

To prove \eqref{eq:Y-evolution}, condition on the benchmark, the auxiliary execution, and the information revealed by the main execution through rank $x$. Consider a request counted by $Y(x)$, and let its auxiliary resource be the unique resource that serves it in the auxiliary execution. If the main execution does not cover this request between ranks $x$ and $y$, then it remains counted at rank $y$ whenever its auxiliary resource has main rank greater than $y$. Since that resource is unprocessed at rank $x$, this happens with conditional probability $(1-y)/(1-x)$. This gives the proportional survival term.

It remains to bound the additional loss caused by the main execution covering such a request at some rank $s\in(x,y]$. The request belongs to $O_P$, so its owner has rank at most $p\le s$ and the capture is counted by $b(s)$. If the main resource that covers the request is also its auxiliary resource, then the auxiliary resource has rank $s\le y$; the request already fails the survival event above and requires no additional charge. Otherwise, for an additional loss to remain at rank $y$, the auxiliary resource must still be unprocessed at rank $s$ and must have main rank greater than $y$. Conditional on being unprocessed at rank $s$, its exact main rank is uniform on $(s,1)$, so this probability is $(1-y)/(1-s)$. Since each request is covered by the main execution at most once, the total additional normalized loss is at most
\[
    (1-y)\int_x^y\frac{b(s)}{1-s}\,ds.
\]
Combining this with the proportional survival term proves \eqref{eq:Y-evolution}.

Apply \eqref{eq:Y-evolution} from $p$ to $x$, use $Y(p)\ge qD(p)$ and $q=1-p$, and then use \eqref{eq:gain-from-Y}. This gives \eqref{eq:rank-cut}. Thus the only loss beyond the proportional disappearance of auxiliary resources is paid for by the same $b$-mass that creates surplus in \cref{cor:baseline}.
\end{proof}

\begin{proof}[Completion of the proof of \cref{thm:main}]
\Cref{lem:owner-inequalities,lem:rank-cut} give three lower bounds on $w$. \Cref{lem:analytic-bound} combines them with nonnegative weights, canceling the $a$-terms and leaving nonnegative $b$-terms. It follows that
\[
    \frac{\E[\ALG]}M=\int_0^1 w(x)\,dx
       \ge\frac{5-2\sqrt3}{3}.
\]
\end{proof}

\section{Weighted Resources}\label{sec:weighted}

We now turn to weighted resources and prove the black-box reduction of \cref{thm:weighted-reduction}.

Let $M=\E[\OPT]$ be the value of the optimal weighted benchmark. Assume that \Ranking{} is $\alpha$-competitive on every unweighted instance of the model for some $\alpha>1/2$. We use the unweighted guarantee only through the value of $\alpha$; the argument below does not rely on how that guarantee is proved. By \cref{thm:main}, we may take $\alpha=\rhostar$.

We start from weighted greedy, which provides a $1/2$-competitive baseline, and use the unweighted guarantee only as a sparse improvement on top of it. Since the unweighted guarantee does not distinguish reward magnitudes, we invoke it only within groups of resources with comparable weights. Most groups continue to follow weighted greedy, while a small random fraction run \Ranking{}.

This creates two sources of loss. An active resource may take a request that weighted greedy would have assigned to an inactive resource; we control this by restricting cross-bucket edges. Also, a request useful to one active bucket may already have been taken by a higher active bucket; sparse activation limits this interference. The analysis below bounds the first loss by the slack of weighted greedy and the second by the reward earned by the active buckets themselves.

\subsection{Weighted-Greedy Baseline}

Fix a strict priority order $\prec$ that lists resources in nonincreasing order of $c_v$, with deterministic tie-breaking. Run an auxiliary weighted-greedy execution with this priority order and its own duration samples. At each request $t$, let $g(t)$ be the resource selected by this execution, with $g(t)=\bot$ if it rejects. Write $\mathcal G_v$ for the schedule of resource $v$, and let
\[
    G=\E\sum_{v\in V}c_v|\mathcal G_v|
\]
be its expected reward.

Run the benchmark independently and use the owner notation from \cref{sec:owners}: $O_v$ is the benchmark schedule of resource $v$, $O=\bigcup_v O_v$, and $o(t)$ is the owner of request $t\in O$. For every request served by both executions, compare the greedy resource $g(t)$ with the benchmark owner $o(t)$ in the priority order $\prec$. Define
\[
    \mathcal E=\{t\in O:g(t)\ne\bot,\ g(t)\prec o(t)\},
    \qquad
    \mathcal L=\{t\in O:g(t)\ne\bot,\ g(t)\not\prec o(t)\}.
\]
For $t\in\mathcal E$, greedy uses a resource of reward at least $c_{o(t)}$; for $t\in\mathcal L$, it uses a resource of reward at most $c_{o(t)}$. Define
\begin{equation}\label{eq:greedy-errors}
    \Delta=\E\sum_{t\in\mathcal E}(c_{g(t)}-c_{o(t)}),
    \qquad
    B=\E\sum_{t\in\mathcal L}c_{g(t)},
    \qquad
    \sigma=2G-M.
\end{equation}
The quantity $\sigma=2G-M$ measures the slack in the $1/2$ guarantee of weighted greedy. The next lemma shows that the two comparison terms $\Delta$ and $B$ are together bounded by this same slack. This is the estimate we will later use to control the filtering loss.

\begin{lemma}\label{lem:greedy-bound}
The weighted-greedy execution satisfies
\begin{equation}\label{eq:greedy-bound}
    \sigma\ge\Delta+B\ge0.
\end{equation}
In particular, $G\ge M/2$.
\end{lemma}
\begin{proof}
For each resource $v$, let
\[
    C_v^G=\bigcup_{u\prec v}\mathcal G_u
\]
be the requests already taken by higher-priority greedy resources before $v$ is processed in the resource-by-resource representation. Applying \cref{lem:residual-comparison} to $v$ and its benchmark schedule $O_v$ gives
\[
    \E[c_v|\mathcal G_v|]
    \ge \E[c_v|O_v\setminus C_v^G|].
\]
A benchmark request $t\in O_v$ is removed from $O_v\setminus C_v^G$ exactly when greedy serves $t$ using some resource $u\prec v$, which is precisely the event $t\in\mathcal E$. Summing over $v$ therefore gives
\[
    G\ge M-A,
    \qquad
    A:=\E\sum_{t\in\mathcal E}c_{o(t)}.
\]

Next split the greedy reward according to whether the served request lies in $\mathcal E$, in $\mathcal L$, or outside $O$. On $\mathcal E$, greedy earns $A+\Delta$ in expectation; on $\mathcal L$, it earns $B$; and let $N\ge0$ be its expected reward on requests outside $O$. Thus
\[
    G=A+\Delta+B+N.
\]
Combining the two displays,
\[
    2G-M\ge G-A=\Delta+B+N\ge\Delta+B.
\]
This proves the lemma.
\end{proof}

Thus $\Delta+B$ is small whenever weighted greedy is close to its worst-case value $M/2$. If weighted greedy has larger slack, that slack can absorb a larger filtering loss. This is the link between the greedy baseline and the filtering rule used below.

\subsection{Random Weight Buckets and the Algorithm}\label{random_weight_buckets}

We now describe how the algorithm departs from weighted greedy on a small random set of weight buckets. The bucket ratio $\gamma$ controls the range of rewards within one bucket, the activation probability $\eta$ controls how often a bucket runs \Ranking{}, and $\theta$ limits how much greedy reward an active resource may displace across buckets. We choose
\[
    1<\gamma<2,
    \qquad 0<\eta<1,
    \qquad \theta=\gamma-1.
\]

\paragraph{Randomly shifted buckets.}
Draw $\zeta\sim\mathrm{Unif}[0,1)$ independently of the weighted-greedy execution and define, for every integer $i$,
\begin{equation}\label{eq:buckets}
    \ell_i=\gamma^{i+\zeta},
    \qquad
    \mathcal B_i=\{v:\ell_i\le c_v<\gamma\ell_i\}.
\end{equation}
For $x>0$, let $\bar x=\ell_i$ for the unique $i$ satisfying $\ell_i\le x<\gamma\ell_i$. Thus every resource $v\in\mathcal B_i$ has bucket index $i(v)=i$ and rounded reward $\bar c_v=\ell_i$. Only finitely many buckets are nonempty.

Without the random shift, two nearly equal rewards could always fall on opposite sides of a fixed bucket boundary. The shift makes the expected loss from such a separation proportional to the difference between the rewards; \cref{lem:rounding} makes this precise.

\paragraph{Active buckets and filtering.}
Independently activate each nonempty bucket by a bit $I_i\sim\mathrm{Bernoulli}(\eta)$. Within each active bucket, draw an independent uniform order of its resources. Active resources are ordered first by decreasing bucket index and then by this within-bucket order. Thus each active bucket has the random resource order required by unweighted \Ranking{}, while higher active buckets have priority over lower ones.

An active resource is allowed to take precedence over weighted greedy, but not on every edge. To limit the greedy reward lost across buckets, for an active resource $v$ we keep the edge $(v,t)$ only if
\begin{equation}\label{eq:eligible}
    g(t)=\bot,
    \qquad\text{or}\qquad i(g(t))=i(v),
    \qquad\text{or}\qquad c_{g(t)}<\theta c_v.
\end{equation}
The three cases correspond to three situations. If greedy rejects, taking the request loses no greedy reward. If the two resources lie in the same bucket, we keep the edge so that \Ranking{} sees the full within-bucket instance. Across different buckets, an active resource may displace greedy only when the greedy reward is less than a $\theta$ fraction of the active reward.

\paragraph{Online execution.}
When request $t$ arrives, first determine $g(t)$ from the current state of the auxiliary weighted-greedy execution. If some available active resource has a kept edge to $t$, assign $t$ to the highest-priority such resource. Otherwise, if $g(t)\ne\bot$ and its bucket is inactive, follow the greedy assignment. If neither case applies, reject $t$.

To determine future choices $g(t)$, the auxiliary greedy execution must also be updated when its current assignment is not followed by the actual algorithm. This is the only place where the sampling oracle is needed. If an inactive greedy assignment is followed, the actual use and the corresponding greedy use share the same realized duration. Every other greedy use receives an independent oracle sample. The complete procedure is summarized in \cref{alg:weighted}.

\begin{algorithm}[H]
\caption{Weighted Algorithm}
\label{alg:weighted}
\begin{algorithmic}[1]
\State Initialize the auxiliary weighted-greedy execution with priority order $\prec$.
\State Sample the random shift, independently activate each nonempty bucket, and draw an independent uniform order within each active bucket.
\For{each arriving request $t$}
    \State Let $g(t)$ be the choice of the auxiliary weighted-greedy execution.
    \State Let $A_t$ be the available active resources whose edge to $t$ is kept by \eqref{eq:eligible}.
    \If{$A_t\ne\varnothing$}
        \State Assign $t$ to the highest-priority resource in $A_t$.
    \ElsIf{$g(t)\ne\bot$ and the bucket of $g(t)$ is inactive}
        \State Assign $t$ to $g(t)$.
    \Else
        \State Reject $t$.
    \EndIf
    \State If an inactive greedy assignment is followed, use the same duration for the actual and greedy uses; for every other greedy use, draw a fresh oracle sample.
\EndFor
\end{algorithmic}
\end{algorithm}

The next lemma verifies that this algorithm is valid and gives the form needed later to analyze the active buckets.

\begin{lemma}\label{lem:weighted-implementation}
The weighted algorithm is feasible and nonanticipatory. Conditional on the weighted-greedy execution, the random shift, the bucket activations, and the within-bucket orders, the active part is fixed-priority greedy on the kept-edge graph.
\end{lemma}
\begin{proof}
Generate the duration samples of the auxiliary greedy execution independently of the shift, activation bits, within-bucket orders, and active-use durations. Whenever the actual algorithm follows an inactive greedy assignment, couple the corresponding greedy use to the same duration; every other greedy use receives an independent oracle sample. Whether a greedy use shares an actual duration or receives an oracle sample is determined before that duration is observed. Hence every greedy use still receives a fresh sample from the correct distribution, and the auxiliary execution has exactly the intended law.

For feasibility, consider an inactive resource. The actual algorithm follows only a subset of its greedy uses, and each followed use has exactly the same busy interval in the two executions. Hence the actual busy intervals of that resource form a subset of its greedy busy intervals. In particular, whenever the algorithm is instructed to follow a greedy assignment, that inactive resource is available in the actual execution. Active resources are assigned only after their actual availability is checked.

Finally, once the weighted-greedy execution and the random shift are fixed, condition \eqref{eq:eligible} determines the kept-edge graph. Fixing in addition the activation bits and within-bucket orders determines a priority order on the active resources. The active part of the algorithm is therefore exactly fixed-priority greedy on this graph.
\end{proof}

By the resource-by-resource representation, the decisions made in a higher active bucket do not depend on the activation bits, within-bucket orders, or active-use durations of lower buckets. Thus, when we later expose active buckets from high to low, the randomness of lower buckets remains unrevealed. We use this in \cref{lem:higher-buckets}.

We also need two elementary properties of the random geometric rounding. Define
\begin{equation}\label{eq:r-chi}
    r(\gamma)=\frac{\gamma-1}{\gamma\ln\gamma},
    \qquad
    \chi(\gamma)=\frac1{\gamma\ln\gamma}.
\end{equation}

\begin{lemma}\label{lem:rounding}
For every fixed $c>0$,
\begin{equation}\label{eq:mean-rounding}
    \E_{\zeta}[\bar c]=r(\gamma)c.
\end{equation}
Moreover, for fixed $x\ge y>0$,
\begin{equation}\label{eq:boundary}
    \E_{\zeta}\bigl[\bar y\,
       \1_{\{x,y\text{ lie in different buckets}\}}\bigr]
    \le\chi(\gamma)(x-y).
\end{equation}
\end{lemma}
\begin{proof}
Let $z$ be the fractional part of $\log_\gamma c-\zeta$. Since $\zeta$ is uniform on $[0,1)$, so is $z$, and $\bar c=c\gamma^{-z}$. Therefore
\[
    \E_{\zeta}[\bar c]
      =c\int_0^1\gamma^{-z}\,dz
      =\frac{\gamma-1}{\gamma\ln\gamma}c.
\]

For the second claim, let $d=\log_\gamma(x/y)$. If $0\le d\le1$, a bucket boundary separates $x$ and $y$ exactly when the fractional part associated with $y$ lies in $[1-d,1)$. Hence
\[
    \E_{\zeta}[\bar y\1_{\mathrm{separated}}]
       =y\int_{1-d}^1\gamma^{-z}\,dz
       =\frac{x-y}{\gamma\ln\gamma}.
\]
If $d\ge1$, the two rewards are always separated, and \eqref{eq:mean-rounding} gives
\[
    \E_{\zeta}[\bar y]
       =r(\gamma)y
       =\chi(\gamma)(\gamma-1)y
       \le\chi(\gamma)(x-y).
\]
\end{proof}

Since $\theta=\gamma-1$,
\[
    \frac{r(\gamma)}\theta=\chi(\gamma).
\]

This identity lets the two cases in the filtering bound use the same coefficient $\chi(\gamma)$.

We now analyze the algorithm by tracking how much rounded benchmark value remains available to active buckets. We start with benchmark requests whose owner buckets are active, subtract the part removed by filtering, and then subtract requests already taken by higher active buckets. The unweighted $\alpha$ guarantee applies to what remains. Finally, we combine the resulting active reward with the greedy reward from inactive buckets.

\subsection{Filtering Loss}

Let
\[
    \overline M_{\mathrm{act}}
      =\sum_{t\in O}\bar c_{o(t)}I_{i(o(t))}
\]
be the rounded benchmark value whose owner bucket is active. Conditional on the benchmark and the random shift, each owner bucket is active with probability $\eta$. Therefore \cref{lem:rounding} gives
\begin{align}
    \E[\overline M_{\mathrm{act}}]
      &=\eta\,\E\sum_{t\in O}\bar c_{o(t)}\notag\\
      &=\eta r(\gamma)\,\E\sum_{t\in O}c_{o(t)}
       =\eta r(\gamma)M.
\label{eq:Mact}
\end{align}

Filtering may delete some of these benchmark owner edges. Let the total rounded value of the deleted edges be
\[
    L_{\mathrm{filter}}
      =\sum_{t\in O}\bar c_{o(t)}I_{i(o(t))}
        \1_{\{(o(t),t)\text{ is not kept}\}}.
\]
The next lemma shows that this loss is bounded by the slack $\sigma$ of weighted greedy, rather than by the full benchmark value.

\begin{lemma}\label{lem:filter-loss}
\begin{equation}\label{eq:filter-loss}
    \E[L_{\mathrm{filter}}]\le\eta\chi(\gamma)\sigma.
\end{equation}
\end{lemma}
\begin{proof}
Fix the independent greedy and benchmark executions. If $g(t)=\bot$, the owner edge is always kept by \eqref{eq:eligible}, so only requests in $\mathcal E\cup\mathcal L$ can contribute to $L_{\mathrm{filter}}$.

First consider $t\in\mathcal E$, and write $u=g(t)$ and $v=o(t)$. Then $c_u\ge c_v$. The owner edge $(v,t)$ is certainly kept if $u$ and $v$ lie in the same bucket. Hence deletion requires a bucket boundary between the two rewards. By \eqref{eq:boundary}, the expected rounded owner value lost to this event, over the random shift, is at most
\[
    \chi(\gamma)(c_u-c_v).
\]
The owner bucket is active independently with probability $\eta$, so summing over $\mathcal E$ gives a contribution at most $\eta\chi(\gamma)\Delta$.

Now consider $t\in\mathcal L$, so $c_u\le c_v$. If $c_u<\theta c_v$, condition \eqref{eq:eligible} keeps the owner edge for every shift, and there is no loss. Otherwise $c_u\ge\theta c_v$. Even if the owner edge is deleted, its expected rounded value is at most
\[
    \E_{\zeta}[\bar c_v]
      =r(\gamma)c_v
      \le\frac{r(\gamma)}\theta c_u
      =\chi(\gamma)c_u.
\]
After multiplying by the activation probability and summing over $\mathcal L$, this contributes at most $\eta\chi(\gamma)B$.

Therefore
\[
    \E[L_{\mathrm{filter}}]
      \le\eta\chi(\gamma)(\Delta+B)
      \le\eta\chi(\gamma)\sigma,
\]
where the last inequality is \cref{lem:greedy-bound}.
\end{proof}

\subsection{Reward from Active Buckets}

After filtering, a benchmark request may still be unavailable to its owner bucket because a higher active bucket has already taken it. For a bucket $\mathcal B_i$, let $H_i$ be the set of requests assigned by strictly higher active buckets, and define
\[
    X
      =\sum_{t\in O}\bar c_{o(t)}I_{i(o(t))}
        \1_{\{(o(t),t)\text{ is kept}\}}
        \1_{\{t\in H_{i(o(t))}\}}.
\]
Thus $X$ is the rounded benchmark value that survives filtering but is unavailable to its owner bucket because a higher active bucket has already taken the request.

Fix an active bucket $\mathcal B_i$. Condition on the random shift, the weighted-greedy execution, the randomness of higher active buckets, and the event $I_i=1$. By \cref{lem:weighted-implementation}, the execution within $\mathcal B_i$ is ordinary unweighted \Ranking{} on the kept edges and the requests outside $H_i$. The within-bucket order and active-use durations remain independent under this conditioning.

To apply the unweighted guarantee, we compare with the benchmark uses owned by $\mathcal B_i$ whose owner edge survives filtering and whose request is not in $H_i$. We realize these uses by simulating the original benchmark internally. Whenever a benchmark use is retained, the comparison policy makes the same use and feeds its realized duration into the simulation. Whenever a benchmark use is deleted, the comparison policy skips it and instead feeds the simulation a fresh independent sample from the same duration distribution.

The internal simulation therefore has the same law as the original benchmark. Since the comparison policy only skips uses, each resource is at least as available as in the simulation, so every retained use is feasible. Each duration is generated only after the corresponding decision, so the policy remains nonanticipatory. Under the conditioning above, the benchmark simulation and its fresh samples are independent of the within-bucket \Ranking{} order and the active-use durations.

Every actual use by a resource in $\mathcal B_i$ earns at least $\ell_i$, while every retained benchmark use owned by this bucket has rounded value exactly $\ell_i$. Across all active buckets, the total rounded value of the retained benchmark uses is
\[
    \overline M_{\mathrm{act}}-L_{\mathrm{filter}}-X.
\]
Let $Z$ denote the total \emph{actual} reward earned by active resources. Applying the unweighted $\alpha$ guarantee separately to each active bucket yields
\begin{equation}\label{eq:rank-sum}
    \E[Z]\ge\alpha\,
        \E[\overline M_{\mathrm{act}}-L_{\mathrm{filter}}-X].
\end{equation}

The next lemma bounds $X$ in terms of the active reward itself. A lower owner bucket contributes to $X$ only if it is active, which occurs with probability $\eta$, and its rounded reward is at most a $1/\gamma$ fraction of the reward of the higher-bucket match.

\begin{lemma}\label{lem:higher-buckets}
\begin{equation}\label{eq:higher-buckets}
    \E[X]\le\frac\eta\gamma\E[Z].
\end{equation}
\end{lemma}
\begin{proof}
Fix the random shift and the weighted-greedy execution, and expose active buckets from high to low. Consider an active match $(u,t)$ made by bucket $\mathcal B_j$. Whether this match occurs, and which resource $u$ makes it, depends only on randomness in buckets of index at least $j$. Suppose the benchmark owner $v=o(t)$ lies in a lower bucket $\mathcal B_i$, with $i<j$. At this point the activation bit $I_i$ has not been exposed and remains Bernoulli$(\eta)$. Moreover,
\[
    \bar c_v=\ell_i\le\frac{\ell_j}{\gamma}\le\frac{c_u}{\gamma}.
\]
Thus, conditional on the higher-bucket execution, this active match contributes at most $(\eta/\gamma)c_u$ in expectation to $X$. Each active match can account for at most one term in $X$. Summing over all active matches proves \eqref{eq:higher-buckets}.
\end{proof}

Substituting \eqref{eq:Mact}, \eqref{eq:filter-loss}, and \eqref{eq:higher-buckets} into \eqref{eq:rank-sum} gives
\[
    \E[Z]\ge\alpha\left(
      \eta r(\gamma)M-\eta\chi(\gamma)\sigma
                         -\frac\eta\gamma\E[Z]\right).
\]
After rearranging, define
\[
    \phi=\frac{\alpha\eta}{1+\alpha\eta/\gamma}
\]
and obtain the active-reward bound
\begin{equation}\label{eq:Z-bound}
    \E[Z]\ge\phi\bigl(r(\gamma)M-\chi(\gamma)\sigma\bigr).
\end{equation}

\subsection{Combining Active and Inactive Buckets}

We now combine the active reward with the greedy reward from inactive buckets. Define
\[
    G_{\mathrm{inact}}
      =\sum_{t:g(t)\ne\bot}c_{g(t)}
        \bigl(1-I_{i(g(t))}\bigr).
\]
Conditioning on the auxiliary weighted-greedy execution and the random shift, each relevant bucket is inactive with probability $1-\eta$. Therefore
\[
    \E[G_{\mathrm{inact}}]=(1-\eta)G.
\]

The actual algorithm follows every such inactive greedy assignment unless the request is taken by an active resource $v$. In that case the two resources lie in different buckets, so the same-bucket clause of \eqref{eq:eligible} does not apply. Since the active edge is kept, necessarily
\[
    c_{g(t)}<\theta c_v.
\]
Thus an active match of reward $c_v$ can displace less than $\theta c_v$ of inactive greedy reward. Each active match displaces at most one such greedy assignment, so pathwise
\[
    \ALG\ge G_{\mathrm{inact}}+(1-\theta)Z.
\]

Taking expectations and applying \eqref{eq:Z-bound}, let
\begin{equation}\label{eq:psi}
    \psi=(1-\theta)\phi
        =(2-\gamma)\frac{\alpha\eta}{1+\alpha\eta/\gamma}.
\end{equation}
Then
\begin{align}\label{eq:weighted-bound}
    \E[\ALG]
      &\ge(1-\eta)G+\psi\bigl(r(\gamma)M-\chi(\gamma)\sigma\bigr)\notag\\
      &=\left(\frac{1-\eta}{2}+\psi r(\gamma)\right)M
       +\left(\frac{1-\eta}{2}-\psi\chi(\gamma)\right)\sigma,
\end{align}
where we used $G=(M+\sigma)/2$.

The only instance-dependent term in this bound is $\sigma$. If its coefficient is nonnegative, then $\sigma\ge0$ can only improve the bound and may be dropped. This gives the following lifting theorem.

\begin{theorem}\label{thm:lifting}
Suppose \Ranking{} is $\alpha$-competitive for some $\alpha>1/2$ in the unweighted model. For $1<\gamma<2$ and $0<\eta<1$, let $\psi$ be as in \eqref{eq:psi}. If
\begin{equation}\label{eq:lifting-condition}
    \frac{1-\eta}{2}\ge\psi\chi(\gamma),
\end{equation}
then
\begin{equation}\label{eq:lifting-ratio}
    \frac{\E[\ALG]}{\E[\OPT]}
       \ge\frac12+\psi r(\gamma)-\frac\eta2.
\end{equation}
\end{theorem}
\begin{proof}
By \cref{lem:greedy-bound}, $\sigma\ge0$. Under \eqref{eq:lifting-condition}, the coefficient of $\sigma$ in \eqref{eq:weighted-bound} is nonnegative. Dropping that term gives
\[
    \E[\ALG]
      \ge\left(\frac{1-\eta}{2}+\psi r(\gamma)\right)M
      =\left(\frac12+\psi r(\gamma)-\frac\eta2\right)M.
\]
\end{proof}

\begin{proof}[Proof of \cref{thm:weighted-reduction}]
The algorithm invokes \Ranking{} only on unweighted residual instances, and the analysis uses those instances only through the guarantee $\alpha$. If necessary, we may weaken the guarantee so that $\alpha\in(1/2,1)$.

Set $\gamma=1+\varepsilon$ and $\eta=\varepsilon/2$. As $\varepsilon\downarrow0$,
\[
    \frac{1-\eta}{2}-\psi\chi(\gamma)
       \longrightarrow\frac{1-\alpha}{2}>0,
    \qquad
    \frac{\psi r(\gamma)-\eta/2}{\eta}
       \longrightarrow\alpha-\frac12>0.
\]
Hence, for all sufficiently small $\varepsilon>0$ depending only on $\alpha$, condition \eqref{eq:lifting-condition} holds and the ratio in \eqref{eq:lifting-ratio} is strictly larger than $1/2$. This proves the black-box reduction.
\end{proof}

\subsection{Numerical Parameters}

For the explicit numerical guarantee in \cref{cor:weighted-main}, choose
\begin{equation}\label{eq:explicit-parameters}
    \alpha_0=\frac{12799}{25000}=0.51196<\rhostar,
    \qquad \gamma=\frac{503}{500},
    \qquad \eta=\frac{29}{5000}.
\end{equation}
The elementary numerical estimates needed below are verified in \cref{app:arithmetic}.

\begin{proof}[Proof of \cref{cor:weighted-main}]
By \cref{thm:main}, $\alpha_0$ is a valid unweighted guarantee. The estimates in \cref{app:arithmetic} give
\[
    r(\gamma)>\frac{997}{1000},\qquad
    \frac{1-\eta}{2}-\psi\chi(\gamma)>0,
    \qquad
    \psi\frac{997}{1000}-\frac\eta2>\frac{34}{10^6}.
\]
Thus \eqref{eq:lifting-condition} holds, and \cref{thm:lifting} yields
\[
    \E[\ALG]\ge\left(\frac12+\frac{34}{10^6}\right)M
               =0.500034\,M.
\]
\end{proof}
\section{Discussion}\label{sec:discussion}

Our results show that a single random priority order suffices to beat $1/2$ for unweighted reusable resources, even with arbitrary resource-dependent i.i.d.\ durations. The key insight is to view \Ranking{} resource by resource: each resource serves a schedule of requests left by higher-priority resources. For deterministic durations, this view turns the algorithm into random-order greedy for a coverage function. For stochastic durations, a comparison of residual schedules preserves the expected-gain inequalities needed for the analysis, while tracking benchmark requests relative to their owners' ranks yields the improvement over $1/2$. Finally, our weighted reduction shows that any such unweighted improvement can also yield a weighted guarantee above $1/2$.

For future works, the main open question is to determine the tight competitive ratio of \Ranking{} for reusable resources, even when all durations are the same deterministic constant. The nonreusable special case gives an upper bound of $1-1/e$~\citep{karp-vazirani-vazirani-1990}, leaving a substantial gap from our guarantee. For weighted resources, it remains open whether one can retain more of the unweighted improvement and whether a guarantee above $1/2$ is possible without sampling access to the duration distributions.

\section*{AI Disclosure}
Generative AI tools, especially GPT 5.6 and GPT 6, assisted with many parts of the mathematical development and exposition. The authors carefully verified the mathematical arguments and revised the exposition for clarity. The authors take full responsibility for the final text and results.

\printbibliography

\newpage

\appendix
\crefalias{section}{appendix}

\section{Combining the Rank Inequalities}\label{app:combine}

This appendix contains the calculation deferred from \cref{sec:rank-cut}. No further property of the matching problem is used here.

\begin{lemma}\label{lem:analytic-bound}
Let $w,a,b$ be nonnegative bounded functions satisfying \eqref{eq:owner-capture}, \eqref{eq:owner-residual}, and \eqref{eq:rank-cut}, with $J$ and $D$ defined by \eqref{eq:J-def} and \eqref{eq:D-def}. Then
\[
    \int_0^1w(x)\,dx\ge\frac{5-2\sqrt3}{3}.
\]
\end{lemma}
\begin{proof}
\emph{Choosing the weights.}
Set
\begin{equation}\label{eq:constants}
    q_0=\frac1{\sqrt3},\qquad p_0=1-q_0,
    \qquad \kappa=2-\sqrt3,\qquad \tau=\sqrt3-1.
\end{equation}
These constants satisfy $0<p_0<q_0<\tau<1$, $1-\tau=\kappa$, and $p_0/q_0=\tau$. Define
\begin{equation}\label{eq:h-def}
    h(s)=
    \begin{cases}
        \dfrac{s}{q_0(1-s)},&0\le s\le p_0,\\[1.5mm]
        \dfrac{p_0}{q_0(1-s)},&p_0<s\le q_0,\\[1.5mm]
        0,&q_0<s\le1,
    \end{cases}
    \qquad
    H(s)=\int_s^1(1-r)h(r)\,dr.
\end{equation}
Then \eqref{eq:D-def} at $p=p_0$ becomes
\begin{equation}\label{eq:D-special}
    D(p_0)=\int_0^{q_0}h(s)a(s)\,ds
             -\frac1{q_0}\int_0^{p_0}b(s)\,ds.
\end{equation}

We multiply \eqref{eq:owner-residual}, \eqref{eq:rank-cut}, and \eqref{eq:owner-capture} by weights $\lambda_1,\lambda_2,\lambda_3$ summing to one. The weights below are chosen to cancel the coefficient of $a(s)$ and keep that of $b(s)$ nonnegative:
\begin{equation}\label{eq:lambda1-def}
    \lambda_1(x)=
    \begin{cases}
       \displaystyle\frac12+\frac{\kappa^2}{6(1-x)^2}
         +\frac\kappa3\left(h(x)-\frac{H(x)}{(1-x)^2}\right),
         &x<\tau,\\[2.5mm]
       \displaystyle\frac13,&x\ge\tau,
    \end{cases}
\end{equation}
\begin{equation}\label{eq:lambda23-def}
    \lambda_2(x)=\frac13\1_{\{x\ge\tau\}},
    \qquad \lambda_3(x)=1-\lambda_1(x)-\lambda_2(x).
\end{equation}
Let
\begin{equation}\label{eq:Lambda-L}
    \Lambda(s)=\int_s^1\lambda_1(x)\,dx,
    \qquad
    L(s)=
    \begin{cases}
        \dfrac\kappa{3q_0},&s<p_0,\\[1.5mm]
        \dfrac\kappa{3(1-s)},&p_0\le s<\tau,\\[1.5mm]
        \dfrac13,&s\ge\tau.
    \end{cases}
\end{equation}

We will verify that, for almost every $s\in[0,1]$,
\begin{align}
    &\lambda_i(s)\ge0\quad(i=1,2,3),
        \qquad \lambda_1(s)+\lambda_2(s)+\lambda_3(s)=1,
        \label{eq:mult-nonnegative}\\
    &\lambda_3(s)-\frac{\Lambda(s)}{1-s}
            +\frac\kappa3h(s)=0,\label{eq:mult-a}\\
    &\lambda_3(s)-L(s)\ge0.\label{eq:mult-b}
\end{align}
Moreover,
\begin{equation}\label{eq:mult-value}
    \int_0^1\lambda_1(x)\,dx
      =\frac12+\frac{\kappa^2}{6}
      =\frac{5-2\sqrt3}{3}.
\end{equation}

\smallskip\noindent\emph{Combining the inequalities.}
Multiply \eqref{eq:owner-residual}, \eqref{eq:rank-cut} at $p=p_0$, and \eqref{eq:owner-capture} by $\lambda_1$, $\lambda_2$, and $\lambda_3$, respectively, and integrate. By \eqref{eq:mult-nonnegative},
\begin{align*}
    W\ge{}&\int_0^1\lambda_1(x)(1-J(x))\,dx\\
    &+\int_{\tau}^1\lambda_2(x)
       \left(D(p_0)-\int_{p_0}^x\frac{b(s)}{1-s}\,ds\right)dx\\
    &+\int_0^1\lambda_3(x)(a(x)+b(x))\,dx.
\end{align*}
Since $\int_0^1\lambda_2(x)\,dx=\kappa/3$, substituting \eqref{eq:D-special} and changing the order of integration gives
\begin{align}\label{eq:combined-bound}
    W\ge{}&\int_0^1\lambda_1(x)\,dx\notag\\
      &+\int_0^1 a(s)
          \left(\lambda_3(s)-\frac{\Lambda(s)}{1-s}
                         +\frac\kappa3h(s)\right)ds\notag\\
      &+\int_0^1 b(s)\bigl(\lambda_3(s)-L(s)\bigr)\,ds.
\end{align}
The $a$-term is zero by \eqref{eq:mult-a}, and the $b$-term is nonnegative by \eqref{eq:mult-b}. Therefore, using \eqref{eq:mult-value},
\[
    W\ge\frac{5-2\sqrt3}{3}.
\]
It remains to check the stated properties of the weights.

\smallskip\noindent\emph{Cancellation of the $a$-terms.}
Values at the finitely many breakpoints do not affect any integral.
Integrating \eqref{eq:h-def} gives
\begin{equation}\label{eq:H-pieces}
    H(s)=
    \begin{cases}
      \displaystyle\kappa-\frac{s^2}{2q_0},&0\le s\le p_0,\\[2mm]
      \displaystyle\frac{p_0(q_0-s)}{q_0},&p_0<s\le q_0,\\[2mm]
      0,&s>q_0.
    \end{cases}
\end{equation}
In particular,
\begin{equation}\label{eq:H0}
    H(0)=p_0-\frac{p_0^2}{2q_0}=\kappa.
\end{equation}
Since $H'(s)=-(1-s)h(s)$ almost everywhere, direct integration of \eqref{eq:lambda1-def} gives
\begin{equation}\label{eq:Lambda-formula}
    \Lambda(s)=
    \begin{cases}
       \displaystyle\frac{1-s}{2}-\frac{\kappa^2}{6(1-s)}
                 +\frac{\kappa H(s)}{3(1-s)},&s<\tau,\\[2.5mm]
       \displaystyle\frac{1-s}{3},&s\ge\tau.
    \end{cases}
\end{equation}
The first branch has derivative $-\lambda_1(s)$ and equals $\kappa/3$ at $s=\tau$, matching the second branch. Substituting \eqref{eq:lambda1-def} and \eqref{eq:Lambda-formula} into \eqref{eq:lambda23-def} gives \eqref{eq:mult-a}.

\smallskip\noindent\emph{Nonnegativity.}
We next check that the weights are nonnegative. On the three intervals $[0,p_0]$, $(p_0,q_0]$, and $(q_0,\tau)$, the respective branches of $\lambda_1$ are increasing. Their derivatives are
\[
    \frac{2(3\sqrt3-5)}{3(1-x)^3},
    \qquad
    \frac{16\sqrt3-27}{9(1-x)^3},
    \qquad
    \frac{\kappa^2}{3(1-x)^3},
\]
all positive. Their upper endpoint values are respectively
\[
    -4+\frac{8\sqrt3}{3},
    \qquad \frac{5\sqrt3}{12},
    \qquad \frac23,
\]
all smaller than one. The first two branches agree at $p_0$, and
\[
    \lambda_1(0)=\frac{2(\sqrt3-1)}3>0,
    \qquad
    \lambda_1(q_0^+)=1-\frac{\sqrt3}{4}>0.
\]
Thus $0<\lambda_1<1$ on $[0,\tau)$. There $\lambda_2=0$ and $\lambda_3=1-\lambda_1\ge0$. On $[\tau,1]$, all three weights equal $1/3$. This proves \eqref{eq:mult-nonnegative}.

\smallskip\noindent\emph{The coefficient of $b$.}
On $[0,p_0]$ the function $L$ is constant and $\lambda_3=1-\lambda_1$ is decreasing. Evaluating at $p_0$ gives
\[
    \lambda_3(s)-L(s)
       \ge6-\frac{10\sqrt3}{3}>0.
\]
On $[p_0,q_0]$, $\lambda_3$ is decreasing and $L(s)=\kappa/(3(1-s))$ is increasing, so
\[
    \lambda_3(s)-L(s)
       \ge\frac12-\frac{\sqrt3}{4}>0.
\]
For $q_0<s<\tau$, put $r=1-s\ge\kappa$. Here $h(s)=H(s)=0$, and
\[
    \lambda_3(s)-L(s)
       =\frac12-\frac\kappa{3r}-\frac{\kappa^2}{6r^2}
       =\frac{(r-\kappa)(3r+\kappa)}{6r^2}\ge0.
\]
For $s\ge\tau$, both $\lambda_3(s)$ and $L(s)$ equal $1/3$. This proves \eqref{eq:mult-b}.

\smallskip\noindent\emph{The constant term.}
Finally, \eqref{eq:Lambda-formula} and \eqref{eq:H0} imply
\[
    \int_0^1\lambda_1(x)\,dx
      =\Lambda(0)
      =\frac12-\frac{\kappa^2}{6}+\frac{\kappa H(0)}3
      =\frac12+\frac{\kappa^2}{6}
      =\frac{5-2\sqrt3}{3}.
\]
This proves \eqref{eq:mult-value}.

\end{proof}

\section{Checking the Weighted Constant}\label{app:arithmetic}

We check the parameters in \eqref{eq:explicit-parameters} using rational bounds. First,
\[
    86603^2-3\cdot50000^2=79609>0
\]
implies $\sqrt3<86603/50000$, and hence
\[
    \rhostar=\frac{5-2\sqrt3}{3}
       >\frac{12799}{25000}=\alpha_0.
\]
Substituting $\alpha_0$, $\gamma=503/500$, and $\eta=29/5000$ into \eqref{eq:psi} gives
\begin{equation}\label{eq:psi-exact}
    \psi=\frac{92789409461}{31530292750000}.
\end{equation}
It remains to check feasibility in \cref{thm:lifting} and the claimed improvement.

Put $z=\gamma-1=3/500$. The bounds
\begin{equation}\label{eq:log-bounds}
    z-\frac{z^2}{2}\le\ln(1+z)
       \le z-\frac{z^2}{2}+\frac{z^3}{3}
\end{equation}
give
\[
    r(\gamma)
       \ge\frac{z}{(1+z)(z-z^2/2+z^3/3)}
       >\frac{997}{1000},
    \qquad
    \chi(\gamma)
       \le\frac1{(1+z)(z-z^2/2)}.
\]
For feasibility, the upper bound on $\chi(\gamma)$ yields
\begin{equation}\label{eq:exact-slack-check}
    \frac{1-\eta}{2}-\psi\chi(\gamma)
       \ge \frac{1-\eta}{2}
          -\frac{\psi}{(1+z)(z-z^2/2)}
       >\frac8{1000}>0.
\end{equation}
For the improvement, the lower bound on $r(\gamma)$ yields
\begin{equation}\label{eq:exact-gain-check}
    \psi r(\gamma)-\frac\eta2
       >\psi\frac{997}{1000}-\frac\eta2
       >\frac{34}{10^6}.
\end{equation}
The two strict rational comparisons follow by substituting the value of $\psi$ in \eqref{eq:psi-exact}. Together they verify \cref{cor:weighted-main}.

\section{Second-Order-Supermodular Welfare}\label{app:submodular}

We show that the same three rank inequalities also apply to the second-order-supermodular setting of \citet{korula-et-al-2018}. This argument is separate from the deterministic coverage reduction: here the objects arriving in random order are items, and greedy assigns each item to an agent. We first define the analogue of benchmark contributions, then derive the same two basic rank inequalities and the same rank-cut inequality as in \cref{sec:stochastic}.

\begin{proposition}\label{prop:submodular}
Greedy is $(5-2\sqrt3)/3$-competitive for online submodular welfare in random order when the valuations are normalized, monotone, and second-order supermodular.
\end{proposition}
\begin{proof}
Write an assignment as a set of item--agent pairs, and let $F(A)$ be its total welfare. For the analysis, we use the same welfare function on sets that may contain more than one pair for the same item. Fix an optimal assignment $O=\{o_1,\ldots,o_n\}$, where $o_i$ assigns item $i$ to its optimal agent, and write $M=F(O)>0$. The normalization means that $F(\varnothing)=0$.

\emph{Contributions of the optimal items.}
Following the gain decomposition of \citet{korula-et-al-2018}, index the items in a fixed order independent of their arrival ranks. Let $O_{<i}=\{o_j:j<i\}$ and $O_{\le i}=O_{<i}\cup\{o_i\}$, and define
\[
    g_i(A)=F(A\cup O_{\le i})-F(A\cup O_{<i}).
\]
Thus $g_i(A)$ is the remaining contribution of the optimal pair $o_i$ after the assignments in $A$ are present. Each $g_i$ is nonnegative and decreasing. Under second-order supermodularity, each $g_i$ is also supermodular. Hence, for any item set $P$,
\[
    R_P(A,S)=\sum_{i\in P}\bigl(g_i(A)-g_i(A\cup S)\bigr)
\]
is a nonnegative monotone submodular function of $S$, with value zero at $S=\varnothing$. It measures how much adding $S$ reduces the remaining contributions of items in $P$.

Telescoping gives $\sum_i g_i(\varnothing)=M$ and, for every pair $e$,
\begin{equation}\label{eq:submod-gain}
    R_{[n]}(A,\{e\})
       \le F(A\cup\{e\})-F(A).
\end{equation}
Indeed, the left side equals the right side minus
$F(A\cup O\cup\{e\})-F(A\cup O)\ge0$.
Also, $g_i(A)\le F(A\cup\{o_i\})-F(A)$ by submodularity, so greedy's gain when item $i$ arrives is at least $g_i(A)$.

Assign independent uniform arrival ranks $U_i$ and let $A_x$ be the greedy assignment through rank $x$. As in \cref{sec:owners}, normalize by $M$ and let $w(x)$ be the welfare-gain density. Split the reduction in the optimal contributions $g_i$ caused by greedy assignments into two densities: $a(x)$ for items whose arrival rank is still larger than $x$, and $b(x)$ for items whose rank is at most $x$. By \eqref{eq:submod-gain}, the welfare gained by greedy dominates the total reduction in these contributions, so
\[
    w(x)\ge a(x)+b(x).
\]
Also, when an unprocessed item $i$ arrives, greedy can gain at least its current contribution $g_i(A)$. The same deferred-rank calculation as in \cref{lem:owner-inequalities} therefore gives
\[
    w(x)\ge 1-J(x),
    \qquad J(x)=\int_0^x\frac{a(s)}{1-s}\,ds.
\]

\smallskip\noindent\emph{Comparison across a rank cut.}
Fix $p\in(0,1)$, put $q=1-p$, and let $P=\{i:U_i\le p\}$ and $Q=[n]\setminus P$. Run greedy from the empty assignment on $Q$ in an independent uniform order, obtaining $\widehat A$. As in the matching proof, embed this auxiliary order on $Q$ into the first $q$ fraction of a fresh full random order. Applying the same deferred-rank probabilities to the reductions in the $g_i$ gives
\begin{align*}
    \frac1M\E R_P(\varnothing,\widehat A)&=pJ(q),\\
    \frac1M\E R_P(\varnothing,A_p)
       &=\int_0^p b(s)\,ds
         +\int_0^p\frac{p-s}{1-s}a(s)\,ds.
\end{align*}
Since the $g_i$ are decreasing,
\[
    R_P(A_p,\widehat A)
       \ge R_P(\varnothing,\widehat A)
           -R_P(\varnothing,A_p).
\]
For $x\ge p$, let $\widehat A_{>x}$ contain the auxiliary assignments of items whose main ranks exceed $x$, and define
\[
    Y(x)=\frac1M\E R_P(A_x,\widehat A_{>x}).
\]
This is the remaining auxiliary reduction in the contributions of the early items $P$. The preceding two displays imply $Y(p)\ge qD(p)$, with $D$ as in \eqref{eq:D-def}.

Condition on $P,Q$, on the greedy execution revealed through rank $x$, and on the auxiliary assignment. For an unprocessed item, greedy's gain when that item is processed is at least the marginal value of its auxiliary assignment, and hence at least the corresponding singleton value of $R_P(A_x,\cdot)$ by \eqref{eq:submod-gain}. Submodularity gives
\[
    \sum_{e\in\widehat A_{>x}}R_P(A_x,\{e\})
       \ge R_P(A_x,\widehat A_{>x}).
\]
The sum on the left is therefore the total comparison value of the remaining auxiliary assignments. Since each remaining item has conditional rank density $1/(1-x)$,
\begin{equation}\label{eq:submod-Y-gain}
    w(x)\ge\frac{Y(x)}{1-x}.
\end{equation}

\smallskip\noindent\emph{Loss as the rank increases.}
As in the matching proof, $Y(x)$ decreases for two reasons: an auxiliary item leaves the set of unprocessed items, or the main greedy assignment reduces the remaining contribution of $P$. The normalization by $1-x$ accounts for the first effect; the second is controlled by $b(x)$.

For the first effect, use the following standard inequality for any nonnegative monotone submodular function $R$ with $R(\varnothing)=0$:
\[
    \sum_{e\in S}\bigl(R(S)-R(S\setminus\{e\})\bigr)\le R(S).
\]
Apply this with $R(S)=R_P(A_x,S)$. Since each remaining item has rank in $(x,x+dx]$ with probability $dx/(1-x)$, removing auxiliary assignments decreases $Y$ at rate at most $Y(x)/(1-x)$.

Adding a main assignment $e$ decreases $R_P(A,S)$ by at most $R_P(A,\{e\})$, because
\[
    R_P(A,S)-R_P(A\cup\{e\},S)
       =R_P(A,\{e\})-R_P(A\cup S,\{e\})
       \le R_P(A,\{e\}).
\]
All items of $P$ have already arrived. These reductions are therefore included in $b(x)$. Combining the two changes gives, almost everywhere,
\[
    Y'(x)\ge-\frac{Y(x)}{1-x}-b(x),
    \qquad
    \left(\frac{Y(x)}{1-x}\right)'
       \ge-\frac{b(x)}{1-x}.
\]
Because the item set is finite, after conditioning on $P$ and $Q$ the relevant expectations are finite polynomials in $(x-p)/(1-p)$; the derivatives above are therefore justified. Integrating from $p$ to $x$ and using $Y(p)\ge qD(p)$ and \eqref{eq:submod-Y-gain} yields
\[
    w(x)\ge D(p)-\int_p^x\frac{b(s)}{1-s}\,ds.
\]
We have therefore obtained the same three inequalities \eqref{eq:owner-capture}, \eqref{eq:owner-residual}, and \eqref{eq:rank-cut} as in the matching proof. The corresponding densities are bounded because the item set is finite, so \cref{lem:analytic-bound} applies and gives the stated ratio.
\end{proof}

\end{document}